\documentclass[12pt]{article}
\usepackage[english]{babel}
\usepackage[T1]{fontenc}

\usepackage[letterpaper,top=2.54cm,bottom=2.54cm,left=2.54cm,right=2.54cm,marginparwidth=2cm]{geometry}

\emergencystretch=\maxdimen
\usepackage{amsmath}
\usepackage{graphicx}
\usepackage[colorinlistoftodos]{todonotes}
\usepackage[colorlinks=true, allcolors=blue]{hyperref}
\usepackage{setspace}
\usepackage[round]{natbib}
\usepackage{amssymb}
\usepackage{amsthm}
\usepackage{caption}
\usepackage{subcaption}
\usepackage{titling}
\newcommand{\subtitle}[1]{%
  \posttitle{%
    \par\end{center}
    \begin{center}\large#1\end{center}
    \vskip0.5em}%
}
\usepackage{hyperref}
\hypersetup{colorlinks,%
citecolor=black,%
filecolor=black,%
linkcolor=black,%
urlcolor=black
}
\usepackage[title]{appendix}
\usepackage{enumitem}
\usepackage{lmodern}
\newtheorem{definition}{Definition}
\newtheorem{conjecture}{Conjecture}
\newtheorem{proposition}{Proposition}

\newtheorem{corollary}{Corollary}
\newtheorem{lemma}{Lemma}

\newtheorem{thm}{Theorem}

\title{Between a Stone and a Hausdorff Space}
\author{Jingyi Wu\thanks{Department of Philosophy, Logic and Scientific Method, London School of Economics. Jingyi.Wu@lse.ac.uk} \\
James Owen Weatherall\thanks{Department of Logic and Philosophy of Science, UC Irvine. James.Owen.Weatherall@uci.edu}\\
\\
Forthcoming in \textit{The British Journal for the Philosophy of Science}}

\begin{document}
\maketitle
\begin{abstract}We consider the duality between General Relativity and the theory of Einstein algebras, in the extended setting where one permits non-Hausdorff manifolds.  We show that the duality breaks down, and then go on to discuss a sense in which general relativity, formulated using non-Hausdorff manifolds, exhibits excess structure when compared to Einstein algebras.  We discuss how these results bear on a class of algebraically-motivated deflationist views about spacetime ontology.  We conclude with a conjecture concerning non-Hausdorff spacetimes with no bifurcate curves.\end{abstract}

\section{Introduction}\label{sec:intro}

In standard textbook treatments of general relativity (GR),\footnote{Such as \citet{Wald}.  See fn. \ref{fn:background} for further references.} possible universes, or regions thereof, are represented by geometric structures: viz. Lorentzian 4-manifolds, or \textit{relativistic spacetimes}, which are smooth four-dimensional manifolds endowed with a smooth metric.  But as first observed by \citet{geroch1972}, a different formalism is also available.  On this alternative, one would drop all mention of a manifold and instead represent possible (regions of) universes with a certain kind of algebraic structure.  These structures, which Geroch called \textit{Einstein algebras}, consist of a commutative ring satisfying certain conditions and endowed with further structure.  Indeed, he argued, the entire theory of GR could be developed using just the structure of Einstein algebras.

When Geroch first introduced Einstein algebras, he suggested they might be a step towards a quantum theory of gravity.  The idea was that within GR, it is implicitly assumed that events can be treated as ``point-like'', in the sense that they can be localized in space and time with arbitrary precision.  But in a quantum theory, one should not expect arbitrarily precise determinables, including locations.  Thus, Geroch reasoned, point-like events, which play an apparently fundamental role in classical GR, were unlikely to persist in a quantum theory of gravity.  The Einstein algebra framework, meanwhile, removes, or at least suppresses, the ``manifold of point-like events'', thereby removing one of the hurdles to constructing a quantum theory.

But philosophers of physics soon became interested in Einstein algebras for a different reason.  Not long after Geroch's paper first appeared, John \citet{earman1977} suggested that Einstein algebras might cast new light on the classic substantivalist/relationalist debate.\footnote{For an overview of the recent debate, see \citet{10.1093/oxfordhb/9780195392043.013.0016}.}   Roughly speaking, the idea was that standard geometric formulations of spacetime theories are ``substantivalist'', in the sense that one first posits a structure representing space and time -- that is, the manifold -- and then one introduces matter by defining fields on that manifold.\footnote{\citet{Field} famously defends this perspective; see also \citet{Friedman}.}  This approach looks like one according to which spacetime is ontologically independent of, and prior to, matter within spacetime.  Einstein algebras, meanwhile, cut out the first step.  One can think of the elements of an Einstein algebra as possible global configurations of a (scalar) matter field, and so on this approach, one begins by positing the possible configurations of some form of matter, and then proceeds to develop the theory from there, without ever needing to introduce a manifold.  Einstein algebras, in other words, look like a natural formalism for ``relationalism''.\footnote{Earman goes on to argue that Einstein algebras are not fully relationalist, or at least not fully Leibnizian, since they do not satisfy the principal of sufficient reason (PSR): God has no grounds for choosing between the many distinct-but-isomorphic Einstein algebras \citep{earman1977,earman1986}.  One's views on the PSR aside, this worry seems chimerical to us, on the grounds that the apparent multiplicity of Einstein algebras derives entirely from how we construct mathematical objects in set theory, and not from a multiplicity of the structures themselves \citep{Bradley+Weatherall, Halvorson+Manchak}.} It seems that Earman had hoped that Einstein algebras would lead to a formulation of GR that has less structure than the standard geometric formulation, and this algebraic approach would be a way to excise the excess structure in the geometry.

More recently, \citet{Rosenstock+etal} took up Earman's suggestion, but with a twist.  They show, using a generalization of the Stone duality theorem, that there is a sense in which the Einstein algebra and standard geometric formalisms of GR are \emph{equivalent}.\footnote{Their arguments are heavily indebted to prior work by \citet{nestruev2002}.  \citet{Rynasiewicz1992} also responded to Earman's work on Einstein algebras by invoking general Stone-type dualities between functions and space.} Roughly, they show that two Einstein algebras are structurally equivalent (i.e. isomorphic) if and only if their corresponding relativistic spacetimes are structurally equivalent (i.e. isometric). In fact, they show, one can reconstruct a unique manifold from a given algebra of possible matter distributions (in the form of smooth scalar fields), and vice versa. They go on to suggest that this result supports a deflationary reading of the substantivalist/relationalist debate, at least as understood by Earman and others in the late 20th century.  Because of the possibility of unique reconstruction, a manifold can be viewed as (nothing but) a representation of possible configurations of matter, where the points of a manifold are just loci of field (dis)agreement, characterizing the ways in which possible matter configurations might be distinct.  Thus, the ``substantivalist'' formalism of relativistic spacetimes does not posit ``more structure'' than the ``relationalist'' formalism of Einstein algebras. The two formalisms are, in this sense, equivalent.

Here we present several results that, while not necessarily undermining this picture, introduce further texture.  The principal observation is that the duality between Einstein algebras and relativistic spacetimes on which \citet{Rosenstock+etal} base their arguments depends on the topology of the manifolds under consideration.  In particular, if one expands the class of relativistic spacetimes to include ones that violate the Hausdorff condition, the duality no longer holds, because now structurally non-equivalent manifolds can have structurally equivalent algebras of smooth scalar fields. In other words, the Algebra to Geometry direction of the duality fails---we can no longer reconstruct a unique manifold given an algebra of smooth scalar fields.  As we go on to show, this breakdown is asymmetric: equivalent manifolds always have equivalent algebras of smooth scalar fields, and thus the Geometry to Algebra direction of the duality continues to hold in this generalized context. We prove several propositions that further probe the character of the breakdown of duality.

We take these results to be of interest for several reasons.  Perhaps the most important is that they cast the substantivalist/relationalist debate in GR in new light.  As we discuss below, our main results, in Section \ref{sec:main}, suggest a strong sense in which the points of (some) non-Hausdorff manifolds must be interpreted as having an ontological status independent of the possible configurations of matter.  This suggests that non-Hausdorff manifolds require a ``substantivalist'' interpretation.  Moreover, as we go on to show, in Section \ref{sec:deflation}, there is a precise sense in which not-necessarily-Hausdorff spacetimes have strictly more structure than their algebraic counterparts.  Thus something like Earman's original vision is realized: in the not-necessarily-Hausdorff context, the geometric formulation of the theory appears to prefer (or require) a substantivalist interpretation; but one can excise this extra substantivalist structure by moving to an algebraic approach.  But despite vindicating Earman's intuitions, the results, and the version of substantivalism at issue here, looks very different from characterizations of substantivalism in the contemporary literature, for reasons we discuss in Section \ref{sec:conclusion}.

Another reason the results are of interest is that several philosophers of physics have recently argued that it is fruitful, for some purposes, to drop the Hausdorff condition.  \citet{Manchak}, for instance, observes that the Hausdorff condition is often considered as a ``physical reasonableness'' condition on relativistic spacetimes, along with other conditions such as requiring that spacetime be globally hyperbolic or hole-free. Following a suggestion from \citet[p. 173]{Hawking+Ellis}, he argues that there are examples of non-Hausdorff spacetimes, such as Misner spacetime with two extensions, that \emph{are} nonetheless physically reasonable---or at least, are not as easily ruled out as one might think.  \citet{Luc}, too, has argued that non-Hausdorff manifolds pass several litmus tests for physical reasonability, and even suggests that we should not require the Hausdorff condition in general when specifying the models of general relativity.  She proposes that they should be interpreted as representing multiple future possibilities compatible with some initial data set \citep[see also][]{Luc+Placek}, thereby defending a version of \citet{Belnap}'s branching spacetime theory against well-known critiques from \citet{EarmanPruning}.\footnote{See \citet{Belnap+etal} for a recent extended discussion of branching spacetimes that draws on this recent work by Luc and Placek.}  We take the results in the present paper to show one potential cost to such proposals: accepting them blocks the deflationary interpretation of spacetime events (and manifolds) suggested by \citet{Rosenstock+etal}.

Finally, these results show a sense in which expanding a modal space -- that is, the space of possible models allowed by a theory -- has consequences for the sensible interpretation of models representing individual possibilities.  In other words, if one adopts a principle of uniformity of interpretation, then once we consider non-Hausdorff manifolds, the physical significance of points in Hausdorff models has to change as well.

The remainder of the paper will proceed as follows.  We begin in Section \ref{sec:prelim} with some preliminary remarks to fix ideas and notation, leading to a statement of the duality result that we take as our starting point.  We then proceed in Section \ref{sec:main} to state and prove the principal results of the paper, showing that the duality breaks down asymmetrically (Propositions \ref{prop:mainresults} and \ref{geo-alg}).  In Section \ref{sec:deflation}, we revisit the deflationary account in light of these results.  This discussion motivates two further results: first, there is a sense in which dropping the non-Hausdorff condition introduces genuinely new, distinctly algebraic possibilities (Proposition \ref{novelty}); and second, that there is a precise sense in which non-Hausdorff manifolds (and, by extension, spacetimes) have more structure than the algebras of smooth functions on them (Proposition \ref{functor}).  Finally, in Section \ref{sec:conjecture}, we state a conjecture that, if true, would lead to a novel hierarchy of ``non-Hausdorfness'' for manifolds.\footnote{See fn. \ref{referee}. Since we proposed this conjecture, an argument has been given that it is true \citep{dougherty2023hausdorff}.}  We conclude in Section \ref{sec:conclusion} by elaborating on some of the themes mentioned above.  Proofs of propositions appear in an Appendix.

\section{Preliminaries}\label{sec:prelim}

In standard approaches to GR, one begins with a smooth manifold $M$, points of which represent events---localized occurrences---in space and time.\footnote{\label{fn:background} For background on the mathematical structure of GR using notation similar to ours, see \citet{Wald} or \citet{malament2012}; for a treatment directed at mathematicians, see, for instance, \citet{ONeill}.  For more on smooth manifolds, also directed at mathematicians, see \citet{Lee}.}  This manifold is assumed to be equipped with a smooth pseudo-Riemannian metric $g_{ab}$ of Lorentzian signature, encoding spatiotemporal relations between events, including the effects of gravitation.  The metric is the principal dynamical field in the theory, and standard treatments generally focus on properties of the metric and, especially, solutions to Einstein's equation, which relates the metric to the distribution of energy-momentum in space and time.  For our purposes, however, the metric will play no role.  We focus just on the manifold.

The most general definition of a (smooth) manifold is that it is a topological space $M$ endowed with a smooth \emph{atlas} $\mathcal{C}$. An atlas, here, is a collection of $n-$charts $(U,\varphi)$, where $U$ is an open subset of $M$ and $\varphi$ is a homeomorphism from $U$ to $\mathbb{R}^n$, for fixed $n$, satisfying certain compatibility and maximality conditions.  This atlas may be seen to induce a ``smoothness structure'' on $M$, in the sense that it determines which maps from $M$ to other manifolds are smooth.\footnote{See \citet[Ch. 1]{Lee} or \citet[Ch. 1]{ONeill}. Note that our approach here differs from that of, say, \citet[\S 1.1]{malament2012} or \citet[Ch. 1]{Wald} in that we assume $M$ is equipped with a topological structure, rather than allowing the atlas to induce a topology.  The reason for this is two-fold.  First, the Hausdorff condition (for example) is naturally understood as a condition on topological spaces, and it is awkward to discuss instead conditions on atlases that induce Hausdorff topologies.  And second, Malament and Wald define ``smoothness'' in a very compact way that, unfortunately, becomes more complicated once one moves to non-Hausdorff spaces, because smooth maps on their definition are not necessarily continuous. (We are grateful to an anonymous reviewer for pointing this out to us.) We prefer to avoid these complications, and so we officially adopt the conventions of \citet{Lee}, though our notation is often closer to that of Malament or Wald.} The number $n$, here, is the \emph{dimension} of the manifold; below we will use the expression ``$n-$manifold'' as shorthand for ``$n$ dimensional manifold''.  When there is no possibility of ambiguity, we will also follow standard practice and suppress reference to atlases, and write instead of a ``manifold $M$''.

In standard presentations, more is often required of manifolds than we have required thus far.  In particular, the topological space $M$ is assumed to satisfy several additional conditions.\footnote{These assumptions are so common that many physicists and mathematical physicists work them into the definition of manifold and then proceed to ignore them. For instance, \citet{Lee}, \citet{Wald}, and \citet{Hawking+Ellis} all assume the Hausdorff and paracompactness conditions; and \citet{malament2012} assumes the Hausdorff condition and, later, the Countable Cover Condition.  Wald's Appendix A  provides a nice discussion of some of these conditions; see also \citet[Appendix]{Manchak}.}  Of particular interest for the present paper is that $M$ is standardly assumed to satisfy the \emph{Hausdorff} condition.
\begin{definition}
    A topological space $X$ is \emph{Hausdorff} if for any points $x,y$, there exist open sets $U$ and $V$ such that $x\in U$, $y\in V$, and $U\cap V = \emptyset$.
\end{definition}
\noindent Other conditions are also standard.  For instance, It is also common to assume that $M$ is \emph{paracompact}.
\begin{definition}
    A topological space $X$ is \emph{paracompact} if any open cover admits a locally finite refinement.\footnote{With less jargon: $X$ is paracompact if, for any collection of open sets $\{O_{\alpha}\}_{\alpha\in A}$ satisfying $\cup_{\alpha\in A}O_{\alpha}=M$, there exists a collection of open sets $\{V_{\beta}\}_{\beta\in B}$ such that (a) for every $\beta\in B$, there exists $\alpha\in A$ for which $V_{\beta}\subseteq O_{\alpha}$, (b) $\cup_{\beta\in B} V_{\beta} = M$, and (c) for every point $p\in M$, there exists an open set $W$ containing $p$ that has non-empty intersection with only finitely many elements of $\{V_{\beta}\}_{\beta\in B}$.}
\end{definition}
\noindent For connected manifolds, in the presence of the Hausdorff condition, paracompactness is equivalent to another condition that is also sometimes assumed (instead), known as \emph{second countability}.\footnote{For a clean proof of the equivalence claim, see \citet{Tanaka}.  Note that in some sources, e.g. \citet{malament2012}, one finds the following ``countable cover'' condition imposed on the atlas associated with a manifold: there exists a countable collection of charts $\{(U_i,\varphi_i)\}_{i\in I}$ such that $\cup_i U_i = M$.  This condition is not a topological condition as stated, but if it holds, then the manifold topology induced by the atlas is second countable.  This result holds independently of other topological conditions, including the Hausdorff condition; it follows from the facts that $\mathbb{R}^n$ is second countable and the manifold topology makes chart maps continuous.}
\begin{definition}
    A topological space $X$ is \emph{second countable} if its topology admits a countable base.\footnote{That is, if there exists a countable collection of open sets $\{O_i\}_{i\in I}$ such that for any open set $U$, there exists $J\subseteq I$ for which $\cup_{i\in J}O_i = U$.}
\end{definition}

The condition we relax in Section \ref{sec:main} is the Hausdorff condition. This means that we will consider smooth manifolds in the general sense described above, without also assuming that the underlying topology is Hausdorff. We will call such manifolds ``not-necessarily-Hausdorff''.  It might be instructive to think about the Hausdorff condition by thinking about when it fails. In that case, we would have points on the manifold that open sets cannot separate. A classic example of a non-Hausdorff space can be constructed from two  lines with their standard manifold structures. We identify every point on these two real lines except for the origin. What we are left with is a real line with two origins. These two origins cannot be separated by any two charts.

As for the other topological conditions, we take the following approach.  Most of the propositions we state and prove below involve existence claims.  The examples we give to prove them generally do (or can) satisfy the paracompactness and second-countability conditions.  In this sense we do not relax those conditions.  Note, however, that Prop. \ref{geo-alg}, which is a universal claim, does not require \emph{any} of the topological conditions we have stated here.  It holds for not-necessarily-Hausdorff manifolds, as stated; but also for not-necessarily-paracompact and not-necessarily-second-countable manifolds.

We now turn to the algebraic side of things.  As we noted in the introduction, our motivation for pursuing the questions discussed here is to explore interpretational issues raised by the Einstein algebra reformulation of GR.  On that approach, the subject matter of GR is represented by algebras satisfying certain further conditions and endowed with further structure that plays the role of a metric and which satisfies dynamical equations.\footnote{For more about Einstein algebras and their physical interpretation, see \citet{geroch1972}, \citet{Heller}, or \citet{Rosenstock+etal}.  Further details about the mathematical structures described here can be found in \citet{nestruev2002}.}  But just as we focused only on manifolds in our discussion above, for present purposes the full mathematics and physics of Einstein algebras will play no role.  Only the algebras themselves will be relevant.

To motivate the algebraic approach, consider a smooth manifold $M$, and let $C^\infty(M)$ denote the collection of smooth maps from $M$ to $\mathbb{R}$.\footnote{For present purposes, the topological considerations above play no role, so we do not assume them.}  This collection turns out to inherit some nice algebraic properties from the structure of $M$. In particular, $C^{\infty}(M)$ forms an Abelian group under pointwise addition, with identity given by the constant function $0$; it forms a commutative (associative and unital) ring over that group under pointwise multiplication, with multiplicative identity given by the constant function $1$; and it forms a real algebra over that ring via pointwise scalar multiplication by constant functions.  Many structures one can define on a manifold $M$ can be re-expressed as structures on $C^{\infty}(M)$.  For instance, a smooth vector field is an $\mathbb{R}$-linear function from $C^{\infty}(M)$ to itself satisfying the Leibniz rule: $\xi(fg) = f\xi(g) + g\xi(f)$ for all $f,g\in C^{\infty}(M)$.  And so on.

This construction provides the basic idea behind the Einstein algebra approach.  One begins with an algebra, which can be represented by $C^{\infty}(M)$ for some manifold $M$.  This algebra is taken to represent possible configurations of (scalar) matter or properties of matter.  Other types of matter fields and possible metrics are represented as tensor fields, which are defined in terms of their action on algebra elements.  One can further define structures such as stress-energy tensors, derivative operators, and curvature tensors on the algebra; and one can state Einstein's equation.  In this way, one can drop all reference to the manifold at all.

How are these algebras related to relativistic spacetimes?  Consider, first, the case of smooth, Hausdorff, paracompact manifolds.  In that case, the relationship between manifolds $M$ and the algebras $C^{\infty}(M)$ associated to them is very strong.  We have the following fundamental result.
\begin{thm}[Generalized Stone]\label{stone}
Given two (smooth, Hausdorff, paracompact) manifolds $M$ and $M'$, $M$ and $M'$ are diffeomorphic if and only if $C^\infty(M)$ and $C^\infty(M')$ are isomorphic.
\end{thm}
\begin{proof}
See \citet{Rosenstock+etal}.
\end{proof}
This Theorem says that if two algebras associated with smooth, Hausdorff, paracompact manifolds are isomorphic, then the manifolds they are associated to are diffeomorphic; and, conversely, if the algebras are \emph{not} isomorphic then the manifolds are not diffeomorphic.  This provides a sense in which the entire structure of a smooth manifold is captured by its algebra of smooth functions.

In fact, even more is true. As \citet{nestruev2002} show and \citet{Rosenstock+etal} discuss, one can identify necessary and sufficient conditions on an algebra $A$ for it to be isomorphic to $C^{\infty}(M)$ for some smooth, Hausdorff, paracompact manifold $M$.  Moreover, given such an algebra one can reconstruct a smooth, Hausdorff, paracompact manifold, unique up to diffeomorphism, that gives rise to that algebra by the construction above.  These results show that one can describe the entire subject matter of GR beginning with purely algebraic structures, satisfying certain conditions and endowed with further structure, in such a way that the standard geometric presentation arises as a certain kind of representation of the algebraic structures as algebras of smooth functions on a particular $M$.

Furthermore, we have a certain precise sense in which GR, formulated in the standard way using Lorentzian manifolds, is theoretically equivalent to a theory that represents spatiotemporal structure using Einstein algebras.\footnote{For a recent overview of the theoretical literature, including this particular notion of equivalence, see \citet{weatherallTE1,weatherallTE2}.  \citet{weatherall2020} describes some of the limitations of this notion of equivalence.} As \citet{Rosenstock+etal} show, one can establish that the relationship between smooth manifolds and their associated algebras is (contravariantly) functorial in the sense that the relationships between manifolds and algebras lifts to a natural map from smooth maps to algebra morphisms, and vice versa; and that the functors determined in this way form a categorical equivalence between a category of smooth, Hausdorff, paracompact manifolds and a certain category of algebras. This equivalence can be extended to spacetimes (i.e., manifolds with metrics) and Einstein algebras, in such a way that it preserves the empirical content of both theories on a natural understanding of what that means.

We now proceed to ask: how much of this picture extends to the case where one considers not-necessarily-Hausdorff manifolds?  In particular, is it the case that if not-necessarily Hausdorff manifolds give rise to isomorphic algebras of smooth functions, then they are diffeomorphic?  As we will see in the next section, the answer to this question is ``no''.

\section{Duality Spoiled}\label{sec:main}
As indicated above, the main results of this section may be summarized as follows.  If we allow for not-necessarily-Hausdorff manifolds, then there exist non-diffeomorphic manifolds with isomorphic algebras of smooth maps.  In other words, if we consider the duality results discussed above in this new context, we find that the Algebra to Geometry direction of the duality breaks down, in the sense that an algebra of smooth functions does not uniquely determine a smooth manifold, even up to diffeomorphism.  However, as we show below, the Geometry to Algebra direction still holds, even in the absence of the Hausdorff condition: if two not-necessarily-Hausdorff manifolds are diffeomorphic, then their algebras of smooth functions are isomorphic.

To establish the first result, we will begin with a construction by which one takes an arbitrary Hausdorff manifold, $M$, and generates a new (non-Hausdorff) manifold, $N$, by ``adding a point''.   We will then show that $M$ and $N$ are not diffeomorphic.  Finally, we show that $C^{\infty}(M)$ and $C^{\infty}(N)$ are nonetheless isomorphic.  The result that there exist non-diffeomorphic not-necessarily-Hausdorff manifolds with isomorphic algebras of smooth functions follows immediately.

We begin with our basic construction. Let $(M,\mathcal{C})$ be a smooth, not-necessarily-Hausdorff $n-$manifold and let $p$ be an arbitrary point of $M$. Let $N$ be the set $M\cup \{p'\}$, where $p'\not\in M$; and let $\mathcal{B}$ be the collection of subsets of $N$ consisting of the open sets of $M$ and all sets of the form $ \{p'\}\cup (O\setminus \{p\})$, where $O$ is open in $M$ and contains $p$. Since $\mathcal{B}$ is a $\pi$-system (i.e., it is closed under finite intersections, since the topology on $M$ is), its closure under arbitrary unions forms a topology for $N$.  Assume $N$ is endowed with this topology.  Observe that every $n-$chart in $\mathcal{C}$ is automatically an $n-$chart on $N$.  Now, for each chart $(U,\varphi)\in \mathcal{C}$ containing $p$ in its domain, let $U'$ be the (open) set $U'=\{p'\}\cup (U\setminus\{p\})$ and let $\varphi':U'\rightarrow\mathbb{R}^n$ be the homeomorphism defined by $\varphi':q\mapsto \varphi(q)$ if $q\neq p'$ and $q\mapsto\varphi(p)$ if $q=p'$.  Then $(U',\varphi')$ is also an $n-$chart on $N$; call it the ``mirror chart'' determined by $(U,\varphi)$.  Finally, let $\mathcal{C}'$ consist of the collection of all $n-$charts on $N$ compatible with both $\mathcal{C}$ and all mirror charts on $N$.  We have the following lemma.

\begin{lemma}\label{lem1}
$(N,\mathcal{C'})$ as just defined is a smooth, non-Hausdorff manifold.  Moreover, it is second countable if $(M,\mathcal{C})$ is; and it is paracompact if $(M,\mathcal{C})$ is.
\end{lemma}

We will call this manifold $(N,\mathcal{C}')$ as ``$M$ with an `additional' $p$'' (and, again, often suppress reference to $\mathcal{C}'$).  Observe, now, that if $M$ is Hausdorff, the manifold that results by adding any ``additional'' $p$ is not diffeomorphic to the manifold with which we began.  Before proving this, we first introduce a definition and another lemma.

\begin{definition}
Let $M$ be a non-Hausdorff manifold. Call points $p_1,p_2\in M$ \emph{witness points} if, given any open sets $O_1$ and $O_2$ containing $p_1$ and $p_2$, respectively, $O_1\cap O_2$ is non-empty.
\end{definition}

\noindent Observe that for any manifold $M$ with an ``additional'' $p$, the points $p$ and the additional point $p'$ are witness points.  We now have the following two results.

\begin{lemma}\label{witnessprop}
If $p_1$ and $p_2$ are witness points of some smooth, non-Hausdorff manifold $M$, then for any smooth scalar field $f: M\to\mathbb{R}$, $f(p_1)=f(p_2)$.
\end{lemma}

\begin{lemma}\label{diffeo}
Let $M$ be a smooth, Hausdorff $n-$manifold, let $p\in M$ be some point, and let $N$ be $M$ with an ``additional'' $p$.  Then $M$ and $N$ are not diffeomorphic.
\end{lemma}

What we have established thus far is that given any smooth Hausdorff manifold, we can construct a non-diffeomorphic, non-Hausdorff manifold from it by adding an ``additional'' point to it.  Note that in Lemma \ref{diffeo}, we assume $M$ to be Hausdorff. We might wonder: does the Lemma hold if $M$ is non-Hausdorff? Not in general. In fact, we have counterexamples of it both when $M$ is non-Hausdorff and non-paracompact,\footnote{Let $M$ be a Hausdorff manifold with $p$ duplicated countably infinite many times. $M$ and $M$ with $p$ duplicated one more time are diffeomorphic.} and when $M$ is non-Hausdorff and paracompact.\footnote{Let $M$ be $\mathbb{R}$ with all natural numbers duplicated once. $M$ and $M$ with an ``additional'' $-1$ are diffeomorphic. Thanks to an anonymous referee for this counterexample.} But for some non-Hausdorff manifolds and some $p$, the Lemma does hold.\footnote{Let $M$ be a Hausdorff manifold with an additional $q$. $M$ and $M$ with an ``additional'' distinct point $p$ are not diffeomorphic.}

The final step in the main result is to establish that any smooth not-necessarily-Hausdorff manifold $M$ has an algebra of smooth functions that is isomorphic to that of $M$ with an ``additional'' point.

\begin{lemma}\label{isoprop}
Let $M$ be a smooth, not-necessarily-Hausdorff manifold and let $N$ be $M$ with some ``additional'' point $p$.  Then $C^{\infty}(M) \cong C^\infty(N)$.
\end{lemma}

The main result follows immediately from the foregoing lemmas.  We have established
\begin{proposition}\label{prop:mainresults}
There exist non-diffeomorphic (smooth, second countable, paracompact) not-necessarily-Hausdorff manifolds with isomorphic algebras of smooth functions.
\end{proposition}
\noindent Thus we have a many-to-one relationship, even up to diffeomorphism, between smooth not-necessarily-Hausdorff manifolds and algebras of smooth functions.  It follows, as claimed, that the Algebra to Geometry direction of the duality fails in this context, as it is hopeless to reconstruct a unique not-necessarily-Hausdorff manifold from an algebra of smooth functions.

Now consider the other direction, from Geometry to Algebra.  It turns out that this direction of the generalized Stone Theorem does hold, even for not-necessarily Hausdorff manifolds.

\begin{proposition}\label{geo-alg}
If not-necessarily-Hausdorff manifolds $M$ and $N$ are diffeomorphic, then $C^\infty(M)\cong C^\infty(N)$.
\end{proposition}

Observe that Prop. \ref{geo-alg} would still hold even if we relaxed second countability or paracompactness, by an essentially identical proof.  It remains an open question, to our knowledge, whether the analogues of Prop. \ref{prop:mainresults} would hold if, instead of dropping the Hausdorff condition, we drop other topological conditions, e.g. whether it holds for smooth Hausdorff manifolds that are not necessarily second countable or paracompact.


\section{Deflation Revisited, and Two More Results}\label{sec:deflation}

The results just presented, taken together, have consequences for the interpretation of GR.  If we adopt the Hausdorff condition, then the deflationary interpretation of spacetime points described in Section \ref{sec:intro} is a viable option for interpreting the points of a manifold representing our actual world: manifolds are (just) ways of encoding the possible configurations of matter; points reflect the ways in which matter configurations can agree or differ.  But if we expand to not-necessarily Hausdorff manifolds, this deflationary interpretation is no longer available---at least not for all models. We apparently cannot take spacetime points to be (nothing but) loci for differentiating possible matter configurations.

Instead, it seems that one must adopt an interpretation according to which at least some points (specifically, pairs of witness points) in non-Hausdorff spacetimes represent locations in space and time whose existence cannot be probed, even in principle, via exploring the possible distributions of matter.  The differences between those manifolds are not reflected in differences between (any) of the matter configurations.  Whatever the differences are, they must be characterized in some other way, such as a brute ontological difference between those points.

As we suggested above, this situation has a distinctly ``substantivalist'' character.  Witness points, on this interpretation, take on an ontological status independent of, and not exhausted by, spatio-temporal relations between matter properties.  It is hard to see how a relationalist or deflationist can make sense of these points.  Indeed, one might argue that if we had empirical evidence that non-Hausdorff manifolds were necessary for physics, it would strike a serious blow to the relationalist and deflationist positions.

How might the algebraically-minded relationalist reply?  One promising possibility would be to turn the arguments above on their head.  One might insist that the physical significance of the spacetime manifold is exhausted by the possible configurations of matter that could be defined on it.  From this perspective, there is no physical difference between universes represented by manifolds with isomorphic algebras of smooth functions.  A smooth manifold $M$ and that manifold with an ``additional'' $p$ (for instance) represent the same facts concerning the possible ways in which matter properties can be distributed.  Since matter fields cannot ``see'' the extra point, it has no physical significance.  It is not really there.

An algebraically-minded relationalist may take this line of thought even further. They might hope, for instance, that expanding the class of possibility to non-Hausdorff manifolds does not add distinct algebraic possibilities. More precisely, one might hope that every non-Hausdorff manifold would be an element of an equivalence class of manifolds with the same algebra of smooth functions; and that each of those equivalence classes would contain one (and, by Theorem \ref{stone}, only one, up to diffeomorphism) Hausdorff representative. If this were true, then expanding the space of models to include non-Hausdorff manifolds would not change the space of physically possible worlds, because the physical content of the additional models would be exhausted by that of the Hausdorff spacetimes.   Non-Hausdorff manifolds would turn out to be superfluous.

This hope is moot, however, because of the following result.

\begin{proposition}\label{novelty}There exists a non-Hausdorff manifold $M$ whose algebra of smooth functions $C^{\infty}(M)$ is not isomorphic to that of any Hausdorff manifold. \end{proposition}

\noindent The example we give is a one-dimensional manifold with two ``branches.'' This manifold has an associated algebra that is non-isomorphic to that of any Hausdorff manifold. In what follows, we will call non-Hausdorff manifolds whose algebras of smooth functions do not admit a representation as functions on any Hausdorff manifold \emph{detectably non-Hausdorff}.

There are several observations to make about this proposition.  First, it follows that one cannot collapse all non-Hausdorff manifolds into Hausdorff manifolds that are physically equivalent in the sense of admitting the same algebra of smooth functions.  Adding non-Hausdorff manifolds to the mix really does increase the space of possibilities, even if the physical significance of a manifold is entirely contained in its associated algebraic structure.

Second, this result has an immediate corollary, concerning the character of the algebras associated with detectably non-Hausdorff manifolds.  Recall that we said above that \citet{nestruev2002} provides necessary and sufficient conditions for a commutative algebra to have a representation as the algebra of smooth functions on a (Hausdorff) manifold.  The details require machinery that would be beyond the scope of this paper, but briefly, there are three conditions: an algebra admits a representation as the algebra of smooth functions on a (Hausdorff) manifold if and only if it is \emph{complete} \citep[\S 3.28]{nestruev2002}, \emph{geometric} \citep[\S 3.7]{nestruev2002}, and \emph{smooth} \citep[\S 4.1]{nestruev2002}.   Thus we have
\begin{corollary}
There exist non-Hausdorff manifolds $M$ for which $C^{\infty}(M)$ is not complete, geometric, and smooth.
\end{corollary}
\noindent The assertion stated here is the corollary to Prop. \ref{novelty}.  But inspection of the conditions shows that it is the smoothness condition that fails, where smoothness is the requirement that there exists a finite or countable open covering $\{U_k\}$ of the dual space of an algebra $\mathcal{A}$ such that all of the restrictions of $\mathcal{A}$ to the elements of the covering are isomorphic to $C^{\infty}(\mathbb{R}^n)$ for some fixed $n$. (The dual space of an algebra $\mathcal{A}$, $|\mathcal{A}|$, consists of all algebra homomorphisms onto $\mathbb{R}$.)  This condition will fail for the dual spaces to the algebras of some non-Hausdorff manifolds.  First, observe that given any manifold $M$, every point $p\in M$ can be associated with an element $\hat{p}\in|C^{\infty}(M)|$, by taking $\hat{p}:f\mapsto f(p)$; and also that, by lemma \ref{witnessprop}, if $M$ is non-Hausdorff and $p,p'$ are witness points, then $\hat{p}=\hat{p}'$.  Now consider, again, the example of the manifold $M$ from the proof of Prop. 3.  Any open cover of the dual space to the algebra of smooth functions on $M$, $|C^{\infty}(M)|$, will include some open set that contains $\hat{p}=\hat{p}'$.  And we claim that the restriction of $C^{\infty}(M)$ to any open set containing that point will fail to be isomorphic to $C^{\infty}(\mathbb{R})$, for just the same reasons that $C^{\infty}(M)$ fails to be isomorphic to $C^{\infty}(\mathbb{R})$.

We do not at present have necessary and sufficient conditions for an algebra to have a representation of smooth functions on a not-necessarily-Hausdorff manifold.  Finding those would be a prerequisite to presenting an autonomous generalization of the theory of Einstein Algebras to include (all) non-Hausdorff manifolds.

Let us return, now, to the motivation for stating Prop. \ref{novelty}.  The question at issue was whether \emph{all} non-Hausdorff manifolds can be associated with some Hausdorff manifold by first passing to the associated algebra of smooth functions and then identifying the (unique) Hausdorff manifold on which that algebra could be represented. That idea failed.  But it is nonetheless true that \emph{some} non-diffeomorphic non-Hausdorff manifolds can be associated with Hausdorff manifolds in this way.

This suggests that once we allow for non-Hausdorff manifolds, something like Earman's original idea \citep{earman1977,earman1989} about the relationship between geometric and algebraic approaches gets traction after all: it would seem that not-necessarily-Hausdorff manifolds have \emph{more structure} than algebras of smooth functions, since there are non-diffeomorphic not-necessarily-Hausdorff manifolds with the same algebras associated with them.  The sense of ``excess structure'' here would be that of \citet{weatherallUG}: there would be two formulations of a theory with the same empirical content, but where there exist distinct, non-equivalent models of one corresponding to a single model of the other.  Moving to the algebraic approach would be a way to eliminate that ``excess structure,” because  isomorphic algebras would correspond to equivalence classes of manifolds, where those equivalence classes would include non-diffeomorphic manifolds.


These ideas can be made precise using the criteria of structural comparison developed by \citet{baezsps} and exported to structural comparisons of physical theories by \citet{weatherallUG,weatherallNG}.\footnote{See \citet{BarrettSPS} for a detailed justification of these criteria in the case of first order theories.}  To do so, we first introduce two categories and then consider a functor between them.\footnote{For readers interested in further background reading on category theory, see \citet{Leinster}.}  Our structural analysis will concern the formal properties of this functor.\footnote{Strictly speaking, the functor we consider will be contravariant, which means it ``reverses'' arrows.  Thus, we will be using an extended version of the criterion for structural comparison that has previously been considered in the philosophy of physics literature.  However, just as \citet{Rosenstock+etal} argued that a categorical duality has much the same significance for questions of theoretical equivalence as a categorical equivalence does, we suggest that the upshot of the analysis does not change for the fact that the functor is contravariant.}  According to the standard of comparison under consideration, the objects of some category \textbf{A} of mathematical gadgets has more \emph{structure} than those of another, \textbf{B}, relative to a choice of functor $F:$\textbf{A}$\rightarrow$\textbf{B}, if $F$ is not \emph{full}, i.e., its induced action on hom sets is not surjective; more \emph{stuff} if $F$ is not \emph{faithful}, i.e., its induced action on hom sets is not injective; and more \emph{properties} if $F$ is not essentially surjective, i.e., surjective up to isomorphism in \textbf{B}.  Conversely, if $F$ is full, faithful, and essentially surjective, then \textbf{A} and \textbf{B} are equivalent (or dual) categories, reflecting that their objects have the same structure.

\citet{weatherallUG} proposes extending this idea to physical theories by considering functors $F$ between categories of models of physical theories, where $F$ is required to preserve the empirical consequences of the models.  One formulation of a theory has more structure than another if there is a functor from the models of the first to the models of the second that preserves empirical consequences, but fails to be full.  Intuitively speaking, this means that there are models that are equivalent in the second category, but their corresponding models are not equivalent in the first category. The key example Weatherall considers is formulations of electromagnetism on Minkowski spacetime using vector potentials and formulations using Faraday tensors, where the former is shown to have more structure than the latter.

We will follow \citet{Rosenstock+etal} closely here---except that we will not worry about metrics or other structure defined on manifolds, or their analogs on algebras.  We will consider just categories of not-necessarily-Hausdorff manifolds and algebras of smooth functions on such manifolds, defined as follows.

First, we consider the category \textbf{nnHMan}, which has as objects not-necessarily-Hausdorff manifolds (of any finite dimension); and has as arrows diffeomorphisms.\footnote{One could of course consider a richer category, with more than just isomorphisms, as do \citet{Rosenstock+etal}.  But any acceptable choice for that category would have diffeomorphisms as isomorphisms; and a duality between that enriched category and a suitable algebraic category would imply a duality between the groupoids we consider here.  Thus if the equivalence of groupoids fails, so too must any stronger duality.} Now consider the category \textbf{nnHAlg}, whose objects are the algebras of smooth functions on the objects of \textbf{nnHMan} and whose arrows are algebra isomorphisms.\footnote{Note that, as discussed above, we do not have an independent characterization of the objects of \textbf{nnHAlg}, but that will not matter for the present discussion.}  Finally, observe that this construction gives rise to a natural contravariant functor, $F:$\textbf{nnHMan}$\rightarrow$\textbf{nnHAlg}, directly analogous to the one considered by \citet{Rosenstock+etal}, that takes each object $M$ of \textbf{nnHMan} to $C^{\infty}(M)$ and takes each arrow $\alpha:M\rightarrow N$ to the algebra isomorphism $\tilde{\alpha}:C^\infty(N)\rightarrow C^\infty(M)$ as defined in the proof of Prop. \ref{geo-alg} in the Appendix.  One can easily confirm that $F$ is a functor, since it preserves identity maps and composition.  Again, by analogy to arguments in \citet{Rosenstock+etal}, this functor would induce a functor that preserves the empirical consequences of the theories were we to extend our analysis to consider spacetimes and Einstein algebras. So its properties are probative for present purposes.

We now have the following proposition.
\begin{proposition}\label{functor}
The functor $F:$\textbf{nnHMan}$\rightarrow$\textbf{nnHAlg} as defined above is neither full nor faithful.
\end{proposition}

This proposition shows that the functor $F$ forgets structure, because it is not full.  Since we expect $F$ to determine a functor from not-necessarily-Hausdorff spacetimes to their associated Einstein algebras that preserves empirical content (assumed to be exhausted by the algebraic structure), it would follow that non-Hausdorff spacetimes have not just more structure, but also \emph{excess} structure. There is another formalization of the theory that captures the same empirical content with less structure.

But the proposition shows more than this.  $F$ also forgets \emph{stuff} because it is not faithful.\footnote{See \citet{nguyen+etal} and \citet{Bradley+Weatherall} for a recent discussion about the significance of forgetting stuff.}  The intuition is that a non-trivial diffeomorphism that permutes a pair of witness points but leaves everything else fixed would map to the identity arrow under the action of this functor. So witness points contribute extra ``stuff'' to non-Hausdorff manifolds, leading to additional symmetries.  Thus, we can think of non-Hausdorff manifolds as exhibiting \emph{both} excess structure and excess ontological commitments. Moving to algebras excises both.


\section{A Tripartite Division}\label{sec:conjecture}

In this section, we turn to a question that arises in light of Prop. \ref{novelty}, and whose significance is informed by Prop \ref{functor}.  To motivate this question, first observe that one consequence of Prop. \ref{novelty} is that there seem to be two classes of non-Hausdorff manifolds: ones, like $M$ with an ``additional'' $p$, whose algebra of smooth functions is isomorphic to that of a Hausdorff manifold; and ones, like the example in the proof of Prop. \ref{novelty}, whose algebra is not.

There is another property shared by some, but not all, non-Hausdorff manifolds that has sometimes been used to classify them.  This is the property of admitting \emph{bifurcate curves}, as first introduced by \citet{hajicek1971bifurcate}.\footnote{See also \citet{hajicek1971causality}.  The bifurcate curves we consider here are ``type 2'' in Hajicek's terminology.  Note that Hajicek does not explicitly require curves to be injective, but his proof of the Theorem in \citep[\S 3]{hajicek1971bifurcate} apparently requires it.  That said, requiring injectivity is not substantive, since any manifold with would-be bifurcate curves that are non-injective (e.g., self-intersecting or coming to rest) also admits bifurcate curves in the present sense (just restrict the domains and reparameterize). We are grateful to an anonymous reviewer for encouraging us to think more about these points.}

\begin{definition}[Bifurcate Curves, \citet{hajicek1971causality}]
    A manifold $M$ has a bifurcate curve iff there exist (smooth, injective) curves $\gamma_i:[0,1]\to M$ for $i=1,2$ and some $t\in(0,1]$ such that $\gamma_1(s)=\gamma_2(s)$ for all $s<t$ and yet $\gamma_1(s)\neq\gamma_2(s)$ for all $s\geq t$.
\end{definition}

A manifold has a bifurcate curve iff there exist two curves that agree up to a point, and disagree at that point (and for all subsequent points).  Whether a non-Hausdorff manifold contains bifurcate curves can be seen as a test for its physical reasonableness.  This is because bifurcate (timelike) curves reflect a kind of indeterminism: a particle traversing such a curve, upon reaching the bifurcation point, could equally well continue in either of two (or more) directions, with nothing in the structure of spacetime distinguishing those possible continuations from one another. The thought, then, is that insofar as it is this sort of branching that makes (some) non-Hausdorff manifolds seem ``unphysical'', then those non-Hausdorff spaces that are free of bifurcate curves may be physically reasonable after all \citep[c.f.][]{Manchak}.

All of the examples of non-Hausdorff manifolds considered thus far in the present paper do admit bifurcate curves.\footnote{For the manifold with an ``additional'' $p$, consider any curve whose endpoint is $p$ and observe that there is another curve agreeing with it everywhere except that its endpoint is $p'$.  That the example in the proof of Prop. \ref{novelty} exhibits such curves, meanwhile, is essential to the proof.}  Since we have seen examples of both detectably non-Hausdorff manifolds and non-detectably-non-Hausdorff manifolds, it follows that admitting bifurcate curves can be neither necessary nor sufficient for detectability.  But what is not clear is how detectability relates to the \emph{absence} of bifurcate curves.  That is: are there non-Hausdorff manifolds with no bifurcate curves that are detectably non-Hausdorff?  Are all such manifolds detectably non-Hausdorff?

We propose the following conjecture.
\begin{conjecture}There exist detectably non-Hausdorff manifolds with no bifurcate curves.
\end{conjecture}
We suggest that the classic example of Misner spacetime with two extensions \citep[pp. 173-4]{Hawking+Ellis} is a likely candidate for a non-Hausdorff manifold that is free of bifurcate curves, but whose algebra of smooth functions has no Hausdorff representation.

Suppose this conjecture holds.\footnote{\label{referee} Since we first made this conjecture, \citet{dougherty2023hausdorff} has given an argument showing not only that this is true, but also that non-Hausdorff manifolds are either detectably non-Hausdorff or with bifurcate curves.}    We would then have a (at least) tripartite division of non-Hausdorff manifolds.  On one extreme, there would be non-detectable non-Hausdorff manifolds.  These would exhibit the most mild non-Hausdorff flavors.  (This is so even though some of them have bifurcate curves!)  One might reasonably argue either that considering such manifolds is harmless; or one might argue that they are otiose because they introduce additional structure that makes no difference to the empirical structure of the theory.  But on either view, the non-Hausdorff character is highly localized, and arguably eliminable. Indeed, the algebraic relationalist might even argue that expanding the space of models of GR to include non-detectable non-Hausdorff manifolds would not truly change the space of physical possibilities contemplated by the theory at all, insofar as the physical significance of a manifold is exhausted by the field configurations it supports.

On the other extreme, we would have detectably non-Hausdorff manifolds with bifurcate curves.  Admitting these manifolds to the space of physical possibility would carry significant costs, such as failures of determinism for particles following (causal) bifurcate curves.  Moreover, these pathological possibilities would be manifest in the algebra of smooth fields on the manifold.

But we would also have another class of non-Hausdorff spacetimes, ones that do not admit such bifurcate curves, but which nonetheless have genuinely novel non-Hausdorff behavior, as reflected in their associated algebras of smooth functions. We do not have a view regarding whether these should count as ``physically reasonable.''  It would certainly be of interest to explore the properties of such manifolds.


Table \ref{table1} summarizes the situation.

\begin{table}[h!]
\centering
\begin{tabular}{|p{.25\textwidth} |p{.1\textwidth} | p{.21\textwidth} | p{.1\textwidth} | p{.16\textwidth} |}
\hline
& $\mathbb{R}$& Misner Spacetime with Two Extensions & $\mathbb{R}\cup \{p\}$ & Branching $\mathbb{R}$\\
\hline
 Hausdorff? & Yes & No & No & No\\
\hline
Free of Bifurcate Curves? & Yes & Yes & No & No\\
\hline
Shares Algebra of Smooth Functions with a Hausdorff Manifold? & Yes & No\footnotemark & Yes & No\\
\hline
\end{tabular}
\caption{A tripartite division of not-necessarily-Hausdorff manifolds.}
\label{table1}
\end{table}

\section{Conclusion}\label{sec:conclusion}
\footnotetext{Recall fn. \ref{referee}.}
This paper has considered whether the duality between smooth manifolds and algebras of smooth functions thereon -- and, by extension, the duality between relativistic spacetimes and Einstein algebras -- generalizes to the not-necessarily-Hausdorff case.  We found that for not-necessarily-Hausdorff manifolds, it is not the case that if two manifolds have isomorphic algebras of smooth functions, then the manifolds must be diffeomorphic.  However, diffeomorphic not-necessarily-Hausdorff manifolds always have isomorphic algebras of smooth functions.

We then showed that there exist non-Hausdorff manifolds whose algebras of smooth functions are not isomorphic to that of any Hausdorff manifold.  It follows that the algebras associated with non-Hausdorff manifolds are not necessarily complete, geometric, and smooth in the sense of \citet{nestruev2002}; thus, expanding the collection of relativistic spacetimes to include ones on not-necessarily-Hausdorff manifolds would lead to a corresponding expansion of the collection of Einstein algebras.  Some non-Hausdorff manifolds exhibit genuinely novel behavior, as viewed from their algebras of smooth functions.

But not all of them do.  Once we consider not-necessarily-Hausdorff manifolds, each smooth algebra of functions on such manifolds is associated with a collection of non-diffeomorphic manifolds.  Thus, one can see the passage from not-necessarily-Hausdorff manifolds to algebras of smooth functions as taking equivalence classes of non-diffeomorphic manifolds.  This idea is reflected in the fact that the natural functor, analogous to the one analyzed by \citet{Rosenstock+etal}, from a certain category of not-necessarily-Hausdorff manifolds to a corresponding category of smooth algebras fails to be full (or faithful).  Thus, there is a sense in which the algebraic approach posits less structure than the geometric approach in the not-necessarily-Hausdorff context.  If the empirical content of not-necessarily-Hausdorff GR is exhausted by the algebras of smooth functions on spacetimes, as apparently is the case for standard GR, then not-necessarily-Hausdorff relativistic spacetimes have excess structure.

Finally, we posed a question about the relationship between our results and the existence of bifurcate curves; and we posed a conjecture that, if true, would establish a hierarchy of degrees of non-Hausdorff behavior.

We will conclude with some brief reflections on the philosophical significance of the arguments here.  Once again, perhaps the most important takeaway is that once we move to not-necessarily-Hausdorff manifolds, something like Earman's original picture of the relationship between manifold-based ``substantivalist'' models and algebra-based ``relationalist'' models of GR is realized after all, in two ways.  First, the manifold approach now seems to require an interpretation on which spacetime points are ontologically prior to, and independent of, configurations of matter, in the strong sense that there can be points that differ from one another, i.e., are distinct individuals, but be such that those differences cannot be reflected in any difference in field valuations.  Matter configurations cannot detect the fact that these points are distinct.  It is hard to make sense of spacetime points in this context other than by adopting a version of substantivalism.  Second, conversely, algebras in this context have \emph{less} structure than manifolds, and so one can excise structure by moving from a geometric formulation to an algebraic one.  Thus, we recover a sense in which the ``relationalist'' algebraic approach is more structurally parsimonious than the ``substantivalist'' geometric approach.

It is worth reflecting on how different the substantivalism-relationalism debate looks when viewed from the present perspective, as compared to how that debate has unfolded in the context of GR since the work of \citet{Friedman}, \citet{Field}, \citet{earman1986}, \citet{Earman+Norton}, and other giants of the last quarter of the 20th century.  To see the point most starkly, consider that it is tempting to call the version of substantivalism currently under consideration ``manifold substantivalism''.  But that expression would not be appropriate given the literature, since ``manifold substantivalism'', as characterized by \citet{Earman+Norton}, denies what is called ``Leibniz equivalence,'' which claims that isometric spacetimes represent the same physical possibilities. The form of substantivalism under consideration here, by contrast, can fully embrace Leibniz equivalence.  Since \citet[522]{Earman+Norton} famously argue that one must either deny Leibniz equivalence or deny substantivalism, the form of substantivalism under consideration here would not count for them as a substantivalist theory at all!  This is so despite the fact that it endorses the claim that spacetime points are individuals whose existence is independent of possible matter configurations.

These considerations suggest a different focus for the debate between the substantivalist and relationalist.  The key issue becomes whether it makes sense to posit additional ``places'' in space and time, even if those places cannot be registered, even in principle, in the space of possible matter configurations.  The substantivalist, here, would say ``yes''; the relationalist ``no''.

Another takeaway is that we now have yet another example of how, by expanding the possibility space of GR, we can no longer take for granted previous results we relied on (c.f. \citet{ManchakTheories}). The duality between relativistic spacetimes and Einstein algebras is only one of the many results that can fail in not-necessarily-Hausdorff spaces.  Perhaps most interestingly, the present case seems to be an example in (non-quantum) physics where the results that fail matter for plausible interpretations of the theory.  Indeed, expanding the space of possible models changes the options for interpreting structures of our actual world.  Suppose we adopt the \emph{prima facie} plausible principle that mathematical structures should be interpreted uniformly across the models of a theory. This principle of uniform interpretation would suggest that if non-Hausdorff spacetimes are possible, then insofar as they do not support deflationism, we have to change our strategies for interpreting spacetime points in all models of the theory---including ones representing our actual world.  Such action-at-a-distance in modal space is common in quantum theories, where there is a close relationship between the determinable properties and the space of (all) possible states.  But it is more difficult to think of similar examples in classical physics.

These considerations, taken together, suggest a novel argument for relationalism---or, at least, for algebraicism.  In the not-necessarily-Hausdorff case, by adopting the principle of uniform interpretation, the geometric approach suggests a non-deflationary, substantivalist interpretation of spacetime points, with excess structure as compared to the algebraic approach. Thus parsimony considerations would advise adopting an algebraic-relationalist view.  While one is not forced to that position by the same considerations in standard (Hausdorff) GR, one might argue that parsimony considerations should take into account not just the specific theory under considerations, but also other ``nearby'' theories.  For instance: one might argue that to properly understand Newtonian gravitation, one should consider not only the theory as classically conceived, but also in relation to other theories, such as GR or teleparallel gravity.  The fact that for ``small'' variations on standard GR, the duality between GR and Einstein algebras breaks down, then, might be taken as evidence that \emph{in general} algebraic-relationalist theories are more parsimonious than substantivalist-geometric ones, and thus one should in general prefer the algebraic variants, even in the special cases, such as Hausdorff GR, where they coincide.

Of course, one might reject the arguments just made, perhaps by rejecting the principle of uniform interpretation that we have advanced.  For instance, \citet{Ruetsche} argues against the viability of what she calls ``pristine interpretations'' for some physical theories, where a pristine interpretation is one that applies uniformly across all possible applications.  But from that perspective, the results above are no less interesting, since they offer a novel example of a case where one might prefer an ``adulterated interpertation'' of a (non-quantum) theory.


\section*{Acknowledgments}
\singlespacing
Thank you to David Malament, JB Manchak, Sarita Rosenstock, and Ellen Shi for many discussions about the material in this paper, with special thanks to Rosenstock for posing the question that led to our formulation of Prop. \ref{functor} and to Manchak for assistance with our discussion of bifurcate curves.  We are also grateful to two anonymous reviewers, whose careful and attentive comments greatly improved the manuscript.  Thank you to members of the Philosophy of Physics Research/Reading Group at UC Irvine for feedback. Wu presented this material to audiences in Paris and Krak\'ow; we are grateful for the helpful comments and questions from both audiences.
\newpage
\begin{appendix}
\section*{Appendix: Proofs of Propositions}\label{sec:appendix}
\noindent\textbf{Lemma \ref{lem1}} $(N,\mathcal{C'})$ as just defined [in Section \ref{sec:main}] is a smooth, non-Hausdorff manifold.  Moreover, it is second countable if $(M,\mathcal{C})$ is; and it is paracompact if $(M,\mathcal{C})$ is.

\begin{proof}
We first show that $N$ is non-Hausdorff. To do so, it suffices to show that every open set containing $p'$ and not $p$ is of the form $\{p'\}\cup (U\setminus \{p\})$ for some open set $U$ in $M$ containing $p$.  But this is immediate, since open sets on $N$ arise as arbitrary unions of elements of $\mathcal{B}$, which was the collection of open sets of $M$ and sets of the form $\{p'\}\cup( O\setminus\{p\})$ of $O$ open in $M$ and containing $p$. The only such unions that contain $p'$ and not $p$ are ones of the form $\left(\bigcup_{\alpha}\{(O_{\alpha}\}\right)\cup\left(\bigcup_{\beta}\{\{p'\}\cup O_{\beta}\setminus\{p\}\right)$, where the $O_{\alpha}$ sets are open in $M$ and do not contain $p$; and the $O_{\beta}$ sets are open in $M$ and do contain $p$.  But any such collection can be rewritten as $\{p'\}\cup (U\setminus \{p\})$, where $U=\left(\bigcup_{\alpha,\beta}\{O_{\alpha}\}\cup\{O_{\beta}\}\right)$ is open in $M$ and contains $p$.  Now let $O$ be any open set in $N$ containing $p$ and not $p'$, and let $O'=\{p'\}\cup (U\setminus \{p\})$ be an arbitrary open set in $N$ containing $p'$ and not $p$.  Then $O\cap O' = (O\cap U)\setminus \{p\}\neq \emptyset$, since $O$ and $U$ are both open sets containing $p$.  Thus $N$ is not Hausdorff.

We next show that $N$ is second-countable if $M$ is.  To see this, suppose that $\{O_n\}_{n\in \mathbb{N}}$  is a countable base for $M$.  Now let $I\subseteq \mathbb{N}$ be the collection of indices such that $p\in O_i$ for each $i\in I$.
 We then define a new collection of open sets, $\{\tilde{O}_i\}_{i\in I}$, where for each $i$, $\tilde{O}_i = (O_i\setminus\{p\})\cup \{p'\}$.  Then $\{O_n\}_{n\in \mathbb{N}}\cup \{\tilde{O}_i\}_{i\in I}$ is a countable base for $(N,\mathcal{C}')$.

Now we show that $N$ is paracompact if $M$ is. Suppose that $\{O_\alpha\}_{\alpha\in A}$ is an open cover of $N$. Let $A'\subseteq A$ be the collection of indices such that $p'\in O_{\alpha'}$ for each $\alpha'\in A'$. We define a new collection of open sets, $\{\tilde{O}_{\alpha'}\}_{\alpha'\in A'}$, where for each $\alpha'$, $\tilde{O}_{\alpha'}=(O_{\alpha'}\setminus\{p'\})\cup\{p\}$. Then $\{O_\alpha\}_{\alpha\in A\setminus A'}\cup \{\tilde{O}_{\alpha'}\}_{\alpha'\in A'}$ is an open cover for $M$. Since $M$ is paracompact by assumption, this open cover admits a locally finite refinement $\{V_\beta\}_{\beta\in B}$. Now let $B'\subseteq B$ be the collection of indices such that for each $\beta'\in B'$, $V_{\beta'}\not\subseteq O_\alpha$ for all $\alpha\in A$. We define a new collection $\{\tilde{V}_{\beta'}\}_{\beta'\in B'}$ where for each $\beta'$, $\tilde{V}_{\beta'}=(V_{\beta'}\setminus\{p\})\cup\{p'\}$. Lastly, pick an arbitrary $O_p\subseteq \{O_\alpha\}_{\alpha\in A}$ with $p\in O_p$. We claim that $F=\{V_\beta\}_{\beta\in B\setminus B'}\cup \{\tilde{V}_{\beta'}\}_{\beta'\in B'}\cup \{O_p\}$ is a locally finite refinement for $\{O_\alpha\}_{\alpha\in A}$. That $F$ is a refinement is immediate. We now show that it is locally finite. For every point $q\in M$, since $M$ is paracompact, there exists an open set $W_q$ containing $q$ that has non-empty intersection with only finitely many elements of $\{V_\beta\}_{\beta\in B}$. We consider two cases: (1) For $q\in M$, since $W_q$ (as an open set in $N$) has non-empty intersection with $\tilde{V}_{\beta'}$ if and only if it has non-empty intersection with $V_{\beta'}$ for each $\beta' \in B$ (otherwise we would have open sets in $N$ that separate $p$ and $p'$), $W_q$ has non-empty intersection with only finitely many elements of $F$. (2) The remaining case is $p'$. We define $W_{p'}=(W_p\setminus\{p\})\cup\{p'\}$. Observe that for each element in $F$, it has non-empty intersection with $W_{p'}$ if and only if it has non-empty intersection with $W_{p}$ (otherwise, again, we would have open sets in $N$ that separate $p$ and $p'$). Together with the fact established in (1), $W_{p'}$ has non-empty intersection with only finitely many elements of $F$.

Finally, we show that $(N,\mathcal{C'})$ is a smooth manifold. Recall that we defined $\mathcal{C}'$ as the collection of all $n-$charts on $N$ compatible with both $\mathcal{C}$ and all mirror charts on $N$. We first show that this set is non-empty, by showing that the elements of the collections consisting of $\mathcal{C}$ and the collection of mirror charts are all pairwise compatible.  Since $\mathcal{C}$ is an atlas, its elements are pairwise compatible by construction.  So we have two other cases to establish: (1) that each mirror chart is compatible with each chart in $\mathcal{C}$; and (2) that mirror charts are pairwise compatible.  We consider these in turn.
\begin{enumerate}
	\item Let $(U_1,\varphi_1)$ be a mirror chart and let $(U_2,\varphi_2)$ be in $\mathcal{C}$. Then $p\not\in U_1$ and there exists a $n$-chart, $(U,\varphi)$ in $\mathcal{C}$, for which $(U_1,\varphi_1)$ mirrors. We know that $(U,\varphi)$ and $(U_2,\varphi_2)$ are compatible because $(M,\mathcal{C})$ is a $n$-manifold. Suppose first that $p\in U_2$.  Then $\varphi_1\circ\varphi_2^{-1}:\varphi_2[U_1\cap U_2]\to \mathbb{R}^2$ is the same map as $\varphi\circ\varphi_2^{-1}: \varphi_2[(U\cap U_2)\setminus \{p\}]\to\mathbb{R}^2$, which is smooth and injective, and for similar reasons $\varphi_2\circ\varphi_1^{-1}:\varphi_1[U_1\cap U_2]\to \mathbb{R}^2$ is smooth and injective. So $(U_1,\varphi_1)$ and $(U_2,\varphi_2)$ are compatible. Now suppose $p\not\in U_2$. Then $U_1\cap U_2=U \cap U_2$, $\varphi_1(U_1\cap U_2)=\varphi(U\cap U_2)$, and $\varphi_2(U_1\cap U_2)=\varphi_2(U\cap U_2)$. It follows that $(U_1,\varphi_1)$ and $(U_2,\varphi_2)$ are compatible, since $(U,\varphi)$ and $(U_2,\varphi_2)$ are compatible,
	\item Now suppose $(U_1,\varphi_1)$ and $(U_2,\varphi_2)$ are both mirror charts. In this case, $p'\in U_1\cap U_2$; and there exist $n$-charts $(\tilde{U_1},\tilde{\varphi_1}),(\tilde{U_2},\tilde{\varphi_2})\in\mathcal{C}$ which $(U_1,\varphi_1)$ and $(U_2,\varphi_2)$ mirror, respectively. But observe that $\varphi_1\circ\varphi_2^{-1}=\tilde{\varphi_1}\circ\tilde{\varphi_2}^{-1}$ and $\varphi_2\circ\varphi_1^{-1}=\tilde{\varphi_2}\circ\tilde{\varphi_1}^{-1}$ are smooth and injective. So $(U_1,\varphi_1)$ and $(U_2,\varphi_2)$ are compatible.
\end{enumerate}
By construction, $\mathcal{C'}$ is also maximal. Thus, $(N,\mathcal{C}')$ is a smooth, non-Hausdorff manifold.

\end{proof}

\noindent\textbf{Lemma \ref{witnessprop}} If $p_1$ and $p_2$ are witness points of some smooth, non-Hausdorff manifold $M$, then for any smooth scalar field $f: M\to\mathbb{R}$, $f(p_1)=f(p_2)$.

\begin{proof}
Suppose, on the contrary, that for two witness points $p_1, p_2\in M$, there exists a smooth scalar field $f:M\to\mathbb{R}$, $f(p_1)\neq f(p_2)$. Then, since $\mathbb{R}$ is Hausdorff, there exist open sets $O_1$ and $O_2$ in $\mathbb{R}$ with $f(p_1)\in O_1$ and $f(p_2)\in O_2$ such that $O_1\cap O_2=\emptyset$. Since $f$ is smooth (and in particular, continuous), $f^{-1}(O_1)\ni p_1$ and $f^{-1}(O_2)\ni p_2$ are open sets in $M$. But since $p_1$ and $p_2$ are witness points, $f^{-1}(O_1)\cap f^{-1}(O_2)\neq\emptyset$. Let $q\in f^{-1}(O_1)\cap f^{-1}(O_2)$. $f$ maps $q$ to two distinct points in $\mathbb{R}$, which is a contradiction since $f$ is a function.

\end{proof}

\noindent\textbf{Lemma \ref{diffeo}} Let $M$ be a smooth, Hausdorff $n-$manifold, let $p\in M$ be some point, and let $N$ be $M$ with an ``additional'' $p$.  Then $M$ and $N$ are not diffeomorphic.

\begin{proof}
 Suppose, on the contrary, that there exists a diffeomorphism $\alpha: N\to M$. Then because $\alpha$ is injective, $\alpha(p)\neq\alpha(p')$. Since $M$ is Hausdorff, there exist open sets $O_1$ and $O_2$ in $M$ with $\alpha(p)\in O_1$ and $\alpha(p')\in O_2$ such that $O_1\cap O_2=\emptyset$. Since $\alpha$ is continuous, $\alpha^{-1}(O_1)\ni p$ and $\alpha^{-1}(O_2)\ni p'$ are both open sets in $N$. Since $p$ and $p'$ are witness points, $\alpha^{-1}(O_1)\cap \alpha^{-1}(O_2)\neq \emptyset$. Let $q\in \alpha^{-1}(O_1)\cap \alpha^{-1}(O_2)$. $\alpha$ maps $q$ to two distinct points in $M$, which is a contradiction since $\alpha$ is a function.\footnote{Thanks to an anonymous reviewer for proposing simplifications of the proofs for Lemmas \ref{witnessprop} and \ref{diffeo}.}
    \end{proof}

\noindent\textbf{Lemma \ref{isoprop}} Let $M$ be a smooth, not-necessarily-Hausdorff manifold and let $N$ be $M$ with some ``additional'' point $p$.  Then $C^{\infty}(M) \cong C^\infty(N)$.
\begin{proof}
We establish the isomorphism explicitly.  For every smooth map $f$ in $C^\infty(M)$, define a map $f':N\to\mathbb{R}$ by setting $f'(q)=f(q)$ for $q\in M$ and $f'(p')=f(p)$. We first show that $S$, the collection of such $f'$ for all $f\in C^\infty(M)$, is an algebra. Because all operations are defined pointwise, the only point to check is $p'$; but that the algebraic properties hold at $p'$ is immediately inherited from the fact that they hold at $p$. Furthermore, elements of $S$ are smooth maps from $N$ to $\mathbb{R}$. To see this, pick an arbitrary $n$-chart $(U,\varphi)$ in the atlas for $N$. If $p'\not\in U$, then $f'\circ\varphi^{-1}=f\circ\varphi^{-1}$ is smooth. If $p'\in U$, then $f'\circ\varphi^{-1}=f\circ\tilde{\varphi}^{-1}$ is smooth, where $(\tilde{U},\tilde{\varphi})$ is the $n$-chart in the atlas on $M$ for which $(U,\varphi)$ is its mirror $n$-chart.  Finally, we show that $S=C^\infty(N)$. To see this, suppose, on the contrary, that there exists a smooth map $h: N\to \mathbb{R}$ such that $h\not\in S$. This means either that $h(p')\neq h(p)$ or else $h_{N\setminus\{p'\}}$ is not a smooth map on $M$.  But the first cannot hold due to Lemma \ref{witnessprop}; and the second cannot hold because any smooth map on a manifold, when restricted to an open subset, determines a smooth map on that subset with its inherited manifold structure.  Therefore, $S=C^\infty(N)$.

Finally, it is easy to verify that the map $g:f\mapsto f'$ for all $f\in C^\infty(M)$ is an isomorphism between $C^\infty(M)$ and $C^\infty(N)$.
\end{proof}

\noindent\textbf{Proposition \ref{geo-alg}} If not-necessarily-Hausdorff manifolds $M$ and $N$ are diffeomorphic, then $C^\infty(M)\cong C^\infty(N)$.

\begin{proof}
Suppose that $M$ and $N$ are diffeomorphic. Then there exists a diffeomorphism $\alpha: M\to N$. From this we know that $\alpha$ and its inverse $\alpha^{-1}$ are both smooth.

Now define $\tilde{\alpha}: C^\infty(N)\to C^\infty(M)$ as follows. For any smooth map $\pi:N\to \mathbb{R}$, $\tilde{\alpha}(\pi)=\pi\circ\alpha$. The map $\tilde{\alpha}$ is well defined because, first, for all $\pi$ in $C^\infty(N)$, $\pi\circ\alpha$ is an element of $C^\infty(M)$ because $\alpha$ is also smooth. Second, if $\pi_1=\pi_2$, then $\pi_1\circ\alpha=\pi_2\circ\alpha$. This is because $\alpha$ is a function.

Furthermore, $\tilde{\alpha}$ admits an inverse $\tilde{\alpha}^{-1}:C^\infty(M)\to C^\infty(N)$ defined as follows. For any smooth map $\chi: M\to \mathbb{R}$, $\tilde{\alpha}{}^{-1}(\chi)=\chi\circ\alpha^{-1}$. $\tilde{\alpha}{}^{-1}$ is well-defined for analogous reasons. $\tilde{\alpha}$ and $\tilde{\alpha}{}^{-1}$ are inverses because for any smooth map $\pi\in C^\infty(N)$, $\tilde{\alpha}{}^{-1}(\tilde{\alpha}(\pi))=\pi\circ\alpha\circ\alpha^{-1}=\pi$.

It now suffices to show that $\tilde{\alpha}$ and $\tilde{\alpha}{}^{-1}$ are homomorphisms. We first show that $\tilde{\alpha}$ preserves addition. For any $p\in M$ and $\pi_1,\pi_2\in C^\infty(N)$, $\tilde{\alpha}(\pi_1+\pi_2)(p)=(\pi_1+\pi_2)\circ\alpha(p)=\pi_1\circ\alpha(p)+\pi_2\circ\alpha(p)=\tilde{\alpha}(\pi_1)(p)+\tilde{\alpha}(\pi_2)(p)$. The second equality holds because $\alpha(p)$ is a point in $N$ and we invoke the additive structure in $C^\infty(N)$. We then show that $\tilde{\alpha}$ preserves scalar multiplication. For any $p\in M$, $a\in\mathbb{R}$, and $\pi\in C^\infty(N)$, $\tilde{\alpha}(a\pi)(p)=(a\pi)\circ\alpha(p)=a(\pi\circ\alpha(p))=a\tilde{\alpha}(\pi)(p)$. Again, the second equality holds because $\alpha(p)$ is a point in $N$ and we invoke the scalar multiplication structure in $C^\infty(N)$.  Finally, we show that $\tilde{\alpha}$ preserves multiplication. For any $p\in M$ and $\pi_1,\pi_2\in C^\infty(N)$, $\tilde{\alpha}(\pi_1\cdot \pi_2)(p)=(\pi_1\cdot\pi_2)\circ\alpha(p)=(\pi_1\circ\alpha(p))\cdot(\pi_2\circ\alpha(p))=\tilde{\alpha}(\pi_1)(p)\cdot \tilde{\alpha}(\pi_2)(p)$. Like before, the second equality holds because $\alpha(p)$ is a point in $N$ and we invoke the multiplication structure in $C^\infty(N)$.

Identical arguments show that $\tilde{\alpha}{}^{-1}$ is a homomorphism.
\end{proof}

\noindent\textbf{Proposition \ref{novelty}}  There exists a non-Hausdorff manifold $M$ whose algebra of smooth functions $C^{\infty}(M)$ is not isomorphic to that of any Hausdorff manifold.

\begin{proof}
Note that it suffices to identify a (non-Hausdorff, connected) $n-$manifold $M$ and a smooth function $\varphi$ on $M$ such that no function with the algebraic properties of $\varphi$ lie in $C^{\infty}(N)$ for any connected Hausdorff $n-$manifold $N$.  (We can restrict attention to connected manifolds because no disconnected manifold has an algebra of smooth functions isomorphic to that of a connected manifold.  Likewise, manifolds of different dimension have non-isomorphic algebras.)  Consider, as $M$, the one-dimensional non-Hausdorff ``branching manifold'', i.e., the manifold constructed by taking two copies of $\mathbb{R}$ and identifying all points $x< 0$, but not identifying points $x\geq 0$.   Now consider a smooth function $\varphi$ on $M$ with the following two properties: first, $d_a\varphi$ is nowhere vanishing; and second, $\varphi$ has two zeros not at $M$'s witness points. (Such a function exists: choose any function $\varphi'$ on $\mathbb{R}$, with standard coordinates, that is everywhere increasing in the $+x$ direction, and which has a zero at some point $x> 0$; and let $\varphi$ be the corresponding function on $M$.  Then $\varphi$ will vanish twice, once on each fork, and it will have nowhere vanishing first derivative relative to any non-vanish vector field.)  Such a function would lie in two distinct maximal ideals of $C^{\infty}(M)$ even though its derivative is nowhere vanishing and thus it is strictly monotonically increasing in one direction.  But clearly no function with those (algebraic) properties can exist on any connected Hausdorff manifold (i.e., up to diffeomorphism, the line or circle), by the mean value theorem.
\end{proof}

\noindent\textbf{Proposition \ref{functor}}  The functor $F:$\textbf{nnHMan}$\rightarrow$\textbf{nnHAlg} as defined above is neither full nor faithful.
\begin{proof}
    To show that the functor fails to be full, it is sufficient to consider any arbitrary Hausdorff manifold $M$, any point $p\in M$, and the manifold $N$, which is $M$ with an ``additional'' $p$.  By Lemmas \ref{diffeo} and \ref{isoprop}, $M$ and $N$ are not diffeomorphic, and so there is no arrow $f:M\rightarrow N$.  But $F(M)=F(N)$ and so $1_{F(M)=F(N)}\in\hom(F(N),F(M))$.  Thus the action of $F$ from $\hom(M,N)$ to $\hom(F(N),F(M))$ fails to be surjective. So $F$ is not full.

    To show that the functor is not faithful, meanwhile, consider the manifold $Q$, which we defined as $\mathbb{R}^4$ with an ``additional'' $p$; and consider the action of $F$ on the diffeomorphism $\beta:Q\rightarrow Q$ that takes $p$ to the additional point $p'$, $p'$ to $p$, and leaves everything else unchanged. That this is a diffeomorphism can be see by considering that it is bijective, and that both it and its inverse, when composed with any smooth function $f:Q\rightarrow \mathbb{R}$, is smooth, since by Lemma \ref{witnessprop}, $f\circ\beta = f\circ\beta^{-1} = f$.  But now observe that, by the same fact, $F(\beta)=F(1_Q)=1_{F(Q)}$, since for every $f\in F(Q)$, $\tilde{\beta}(f)=f$.  Thus the action of $F$ from $\hom(Q,Q)$ to $\hom(F(Q),F(Q))$ fails to be injective. So $F$ is not faithful.
\end{proof}

\end{appendix}
\newpage
\bibliographystyle{agsm}
\bibliography{hausdorff}
\end{document}